\documentclass[12pt,a4paper]{article}
\usepackage{graphicx}
\usepackage[]{amsmath,amssymb,amsthm}
\usepackage{paralist}
\usepackage{epstopdf}

\usepackage{tikz}
\usetikzlibrary{decorations.pathreplacing,decorations.markings,positioning}
\usepackage{subfigure}
\usepackage{hyperref}

\tikzset{
    middlearrow/.style n args={3}{
        draw,
        decoration={
            markings,
            mark=at position 0.5 with {
                \arrow[scale=2]{#1};
                \path[#2] node {$#3$};
            },
        },
        postaction=decorate
    }
}

\newtheorem{theorem}{Theorem}[section]

\newtheorem{lemma}[theorem]{Lemma}

\theoremstyle{definition}

\theoremstyle{remark}

\begin{document}

\begin{center}
{\large{\bf Learning coordination through new actions}}\\
\mbox{} \\
\begin{tabular}{c}
{\bf Sofia B.\ S.\ D.\ Castro$^{\dagger}$} \\
{\small (sdcastro@fep.up.pt)} 
\end{tabular}

\end{center}

\noindent $^{\dagger}$ Faculdade de Economia, Centro de Economia e Finan\c{c}as, and Centro de Matem\'atica, Universidade do Porto, Rua Dr.\ Roberto Frias, 4200-464 Porto, Portugal; fax: +351 225 505 050; phone: +351 225 571 100.

\begin{abstract}
We provide a novel approach to achieving a desired outcome in a coordination game: the original $2\times 2$ game is embedded in a $2\times 3$ game where one of the players may use a third action. For a large set of payoff values only one of the Nash equilibria of the original $2\times 2$ game is stable under replicator dynamics. We show that this Nash equilibrium is the $\omega$-limit of all initial conditions in the interior of the state space for the modified $2\times 3$ game. Thus, the existence of a third action for one of the players, although not used, allows both players to coordinate on one Nash equilibrium. This Nash equilibrium is the one preferred by, at least, the player with access to the new action. 

This approach deals with both coordination failure (players choose the payoff-dominant Nash equilibrium, if such a Nash equilibrium exists) and miscoordination (players do not use mixed strategies). 
\end{abstract}


\medskip

\vspace{.3cm}

\section{Introduction}\label{sec:intro}
It is widely accepted that coordination is difficult to achieve in a coordination game. 
Classic examples include the Stag Hunt and the Battle of the Sexes games which we use as illustrations. Let the actions in the Stag Hunt game be to hunt ``Stag'' or ``Hare'', in this order. In the Battle of the Sexes game let the available actions be to go to ``Football'' or the ``Opera'', in this order. The example payoff matrices for these games are represented in Table~\ref{tb:games}.

\begin{table}
\begin{tabular}{c|c|c|}
 & Hare & Stag \\
 \hline 
Hare & $5$, $5$ & $8$, $1$ \\
\hline 
Stag & $1$, $8$ & $10$, $10$ \\
\hline
\end{tabular}
\hspace{4cm}
\begin{tabular}{c|c|c|}
 & Football & Opera \\
 \hline
 Football & $10$, $7$ & $0$, $0$ \\
 \hline
 Opera & $0$, $0$ & $7$, $10$\\
 \hline
\end{tabular}
\caption{The payoff matrices for the games of Stag Hunt (left) and Battle of the Sexes (right). As usual, player $A$ is the row player and player $B$ the column player. Payoffs to player $A$ appear first and to player $B$, second. \label{tb:games}}
\end{table}
But what if one, and only one, of the players could choose a third action, say ``Partridge'' in the Stag Hunt game and ``Cinema'' in the Battle of the Sexes? We show that the player with the additional action can force coordination on their preferred pure equilibrium of the original game. Our approach is fundamentally different than those in the literature.

Typically, a coordination game has multiple Nash equilibria leading to a problem of choice. This is particularly relevant when one of the equilibria provides a better payoff while another is less risky in terms of loss of payoff. The former equilibrium is called {\em payoff-dominant} while the latter is called {\em risk-dominant}. The choice of the actions leading to the risk-dominant equilibrium, providing a smaller payoff, is often referred to as a {\em coordination failure} \cite{Kets_etal2022}\footnote{There are naturally many more references than is reasonable to provide in a single article. The present choice does not reflect any judgement on the references not included.} and many different approaches have been taken to avoid coordination failure. One severe obstacle to achieving coordination on the payoff-dominant outcome is that, as Weidenholzer \cite{Wei2010} clearly states ``no equilibrium refinement concept can discard a strict Nash equilibrium''.

In experiments, Brandts and Cooper \cite{BraCoo2006} use incentives as coordinating devices, while Battalio {\em et al.} \cite{Battalio_etal2001} introduce an ``optimization premium'' and show that a larger optimization premium promotes coordination on the risk-dominant equilibrium (thus coordination failure persists). Also in the context of experiments, Blume and Ortmann \cite{BluOrt2007} show that {\em cheap talk}, that is, costless pre-play comunication, facilitates coordination. Cooper {\em et al.} \cite{Cooper_etal1992} distinguish one- and two-way communication to show that pre-play communication can lead to coordination but its effect varies with the type of game. Another type of cheap-talk, free-form communication, is addressed by Dugar and Shahriar \cite{DugSha2018} to show that if cheap-talk allows for an explanation of the players' choice then it enhances coordination. In the context of learning dynamics, Zhang and Hofbauer \cite{ZhaHof2016} use the notion of a ``quantal response equilibrium'' to show that under replicator dynamics, players may choose to cooperate if there is a strong enough punishment. Peer-punishment, when heterogenous, adds the issue of the effect of the composition of the group of players. Albrecht and Kube \cite{AlbKub2018} show that peer-punishment is indeed heterogeneous in a coordination game.

Other, mostly experimental alternatives, include the introduction of advice from a non-player or a delay. The former only helps in achieving coordination if the players believe the advice not to be self-interested as reported by Kuang {\em et al.} \cite{Kuang_etal2007}. The effect of a delay is, according to Jin {\em et al.} \cite{Jin_etal2023}, that of signalling that the player who is delaying their action is willing to choose the pay-off dominant outcome when playing next.

Even so the conclusion that ``payoff-dominant equilibrium is an extremely unlikely outcome, either initially or in repeated play'' is put forward by Van Huyck {\em et al.} \cite{VanHuyck_etal1990} after studying several equilibrium refinements in experiments.

Some authors attribute the cause of coordination failure to strategic uncertainty or the lack of common knowledge. Strategic uncertainty enhances risk and promotes coordination failure in experiments by Dal B\'o {\em et al.} \cite{DalBo_etal2021}. Riechmann and Weimann \cite{RieWei2008} claim that strategic uncertainty can be minimized by the introduction of a mechanism of competition that, according to experimental data, leads to coordination. In turn, Myatt {\em et al.} \cite{Myatt_etal2002} recognize common knowledge as a stringent assumption and propose the use of a global game, characterized by incomplete information with a type space determined from a noisy signal, for which there is a unique Nash equilibrium.
The existence of unequal payoffs for the Pareto superior equilibrium is also mentioned as a cause for coordination failure, especially when the difference between payoffs increase. This was tested by Chmura {\em et al.} \cite{Chmura_etal2005}.

In coordination games, any of the pure strategy Nash equilibria is sometimes replaced by the choice of a Nash equilibrium using mixed strategies. This is referred to as {\em miscoordination} and is common when the players are allowed to learn by playing over time. Under some form of fictitious play this is shown to be the case by both Fudenberg and Kreps \cite{FudKre1993} and Ellison and Fudenberg \cite{EllFud2000}.

\bigbreak

The present article proposes a new way of achieving coordination in the payoff-dominant equilibrium by using learning in the form of replicator dynamics and the addition of a new action for only one of the players. 
Solving a coordination problem between two agents by providing one of them with an additional action may seem like cheating. However, agents learn not only by improving their use of the originally available actions but also by resorting to additional actions. It is this latter form that is addressed in this article.

The addition of an action has been used by He and Wu \cite{HeWu2020} where the new action represents a compromise, and by Heifetz {\em et al.} \cite{Heifetz_etal2013} to allow one of the players to achieve their most desirable outcome. The elimination of the use of an action has also been considered by Flesch {\em et al.} \cite{Flesch_etal2011}. In all instances the new action is available to all the players and thus totally differs from the present approach. Note that the heterogeneity of players as well as of the environment make it natural that new actions may not be available to all players.

The use of an additional action may not always be feasible, at least not within a reasonable time frame. This is the case if an additional action requires building of infrastructure such as roads or bridges. Nevertheless, it is often the case that an additional action is readily available: in the Battle of the Sexes, a third type of entertainment leading to a third action is surely available.

The work that is presented here uses concepts of dynamical systems 
to establish that if one of the players in a $2 \times 2$ coordination game can use a third action, then there exists an open set of payoffs for this new action that make the most favourite action of this player the only stable outcome of the game, under replicator dynamics.

To be precise, consider a $2 \times 2$ coordination game where players are $A$ and $B$ with available actions, respectively, $\{ A_1, A_2 \}$ and $\{ B_1, B_2 \}$. 
Unless the game is of pure coordination, we assume that the payoff at each Nash equilibrium is different for both players. Then each player has a favourite\footnote{In pure coordination games our approach is still valid if, for some reason external to the game, one of the outcomes is preferred by one of the players.} Nash equilibrium in the sense that their payoff is bigger at one of the pure strategy Nash equilibria. In the Battle of the Sexes players have different favourite outcomes whereas in the Stag Hunt players have the same favourite. Considering that one of the players acquires the use of a third action, we describe the circumstances under which the favourite Nash equilibrium for this player is the outcome of the game. 

The chance to use an additional action may be imposed by an outside planner or may be the result of the deliberate effort of this one player. In the latter case, this is interpreted as learning on the part of this player. For the sake of argument, let player $B$ be the player who can use a third action\footnote{Analogous results can be obtained for player $A$, should they be the one with access to a third action.}, say $B_3$.
The circumstances conditioning  the favourite Nash equilibrium for player $B$ depend on the payoffs atributed to the outcomes where the new action is used, for both players. This allows for some strategic behaviour in the construction of these payoffs, either by the outside planner or by the player, which we do not address here. We show also that player $A$, while using only the two original actions, and although unable to determine the outcome of the replicator dynamics, has yet some strategic power: by perceiving a different payoff for the use of $A_1$ or $A_2$ when confronted with the new action $B_3$ of player $B$, player $A$ can prevent the preferred equilibrium from being unique, thus perpetuating the issue of coordination. These strategic considerations go beyond the game theoretic conclusions of this article. 

In the next section, we construct the modify the coordination game to allow one of the players the use of a third action. Section~\ref{sec:choice} establishes the existence of coordination in the modified game and the final section concludes.

\section{A modification of a coordination game}\label{sec:CoordinationGames}

Consider a $2 \times 2$ coordination game where players are $A$ and $B$ with available actions, respectively, $\{ A_1, A_2 \}$ and $\{ B_1, B_2 \}$. 
Denote by $a_{ij}$ the payoff received by player $A$ by using action $A_i$ against action $B_j$ of player $B$. Denote the payoff of this interaction received by player $B$ by $b_{ij}$. We have $a_{11}-a_{21}>0$, $a_{22}-a_{12}>0$, $b_{11}-b_{12}>0$, and $b_{22}-b_{21}>0$ to ensure that both $(A_1, B_1)$ and $(A_2, B_2)$ are pure strategy Nash equilibria. 
The matrix of payoffs is 
$$
\left( \begin{array}{cc}
a_{11},b_{11} & a_{12},b_{12} \\
\mbox{} & \mbox{} \\
a_{21},b_{21} & a_{22},b_{22}
\end{array} \right).
$$

Let player $B$ learn to use a new action, say $B_3$. The matrix of payoffs\footnote{We do not address ways of constructing particular payoffs in order to preserve generality.} becomes
$$
\left( \begin{array}{ccc}
a_{11},b_{11} & a_{12},b_{12} & a_{13},b_{13} \\
\mbox{} & \mbox{} \\
a_{21},b_{21} & a_{22},b_{22} & a_{23},b_{23}
\end{array} \right).
$$
Assume that either $b_{13}>b_{11}$ or $b_{23}>b_{22}$ so that player $B$ has some benefit from the use of the new action $B_3$.
If both $b_{13}>b_{11}$ and $b_{23}>b_{22}$ then there is only one Nash equilibrium for the one stage game which will consist in player $B$ using the new action $B_3$ and player $A$ choosing depending on the relative magnitude of $a_{13}$ and $a_{23}$. This is not relevant for our problem which is that of achieving coordination on one of the two initially available actions. 

Note that if $b_{13}>b_{11}$ then the $2 \times 3$ stage game has only one Nash equilibrium, $(A_2, B_2)$, provided $a_{23}>a_{13}$. Otherwise, an additional Nash equilibrium appears: $(A_1, B_3)$. 
On the other hand, if $b_{23}>b_{22}$ then either only $(A_1, B_1)$ is a Nash equilibrium of the $2 \times 3$ stage game provided $a_{23}<a_{13}$, or an additional Nash equilibrium appears : $(A_2, B_3)$. Recall that the existence of a unique Nash equilibrium for the stage game does not guarantee convergence to it when dynamics is introduced. The Rock-Scissors-Paper game provides perhaps the most famous example.

The replicator dynamics is governed by the following ODEs, where we write $a_{ijkl}=a_{ij}-a_{kl}$ as well as $b_{ijkl}=b_{ij}-b_{kl}$,

\begin{equation}
\left\{ 
\begin{aligned}
\dot{x} & = x(1-x)\left[ a_{1121}y_1 + a_{1222}y_2 + a_{1323}(1-y_1-y_2) \right]  \\
& \mbox{} \\
\dot{y}_1 & =y_1\left\{ \left[ b_{2123} - (b_{2123}+b_{1311})x \right](1-y_1) - \left[ b_{2223}- (b_{2223}+b_{1312})x \right]y_2 \right\}  \\
& \mbox{} \\
\dot{y}_2 & =y_2\left\{ \left[ (b_{2223} - (b_{2223} +b_{1312})x \right](1-y_2) - \left[ (b_{2123} - (b_{2123}+b_{1311})x \right]y_1 \right\} 
\end{aligned}
\right.
\label{eq:ODE}
\end{equation}

\medskip

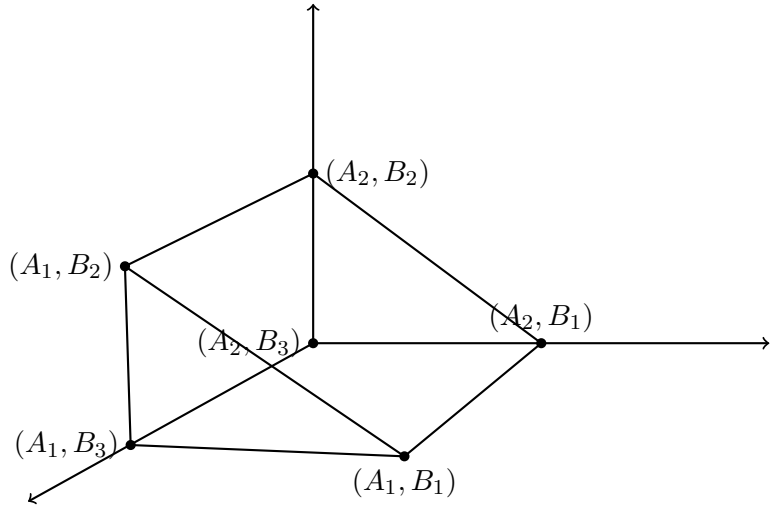
\begin{figure}[!htb]
\centering
\begin{tikzpicture}[>=latex,thick,scale=1.5]
\coordinate (1) at (0,0);
\coordinate (2) at (4,0);
\put (175,0,0) {$y_1$};
\coordinate (3) at (0,3);
\put (1,135,0) {$y_2$};
\coordinate (4) at (-2.5,-1.4);
\put (-115,-55) {$x$};
\coordinate (5) at (2,0);
\coordinate (6) at (0,1.5);
\coordinate (7) at (-1.6,-.9);
\coordinate (8) at (.8,-1);
\coordinate (9) at (-1.65,.68);
\path (1) edge[-to] (2);
\path (1) edge[-to] (3);
\path (1) edge[-to] (4);
\filldraw[black] (1) circle (1pt) node[left] {\small$(A_2,B_3)$};
\filldraw[black] (2,0) circle (1pt) node[above] {\small$(A_2,B_1)$};
\filldraw[black] (0,1.5) circle (1pt) node[right] {\small$(A_2,B_2)$};
\filldraw[black] (-1.6,-.9) circle (1pt) node[left] {\small$(A_1,B_3)$};
\filldraw[black] (.8,-1) circle (1pt) node[below] {\small$(A_1,B_1)$};
\filldraw[black] (-1.65,.68) circle (1pt) node[left] {\small$(A_1,B_2)$};
\path (5) edge[] (8);
\path (7) edge[] (8);
\path (5) edge[] (6);
\path (8) edge[] (9);
\path (7) edge[] (9);
\path (6) edge[] (9);
\end{tikzpicture}
 \caption{The state space of the modified game is the wedge contained in $[0,1]^3$. The vertices are all equilibria for the dynamics of \eqref{eq:ODE} and are labelled according to the action used by each player.}
\label{fig:StateSpace}
\end{figure}

The phase space is $S=\{ (x,y_1,y_2): \; \; x,y_1,y_2 \in [0,1]  \mbox{ and  } y_1+y_2 \leq 1 \}$. The corners of $S$ are equilibria for the dynamics. See Figure~\ref{fig:StateSpace}. We define the following flow invariant planes:
\begin{itemize}
	\item  $S_{y_1} = \{ (x,0,y_2) \} \subset S$, where player $B$ does not use the action $B_1$, 
	\item  $S_{y_2} = \{ (x,y_1,0) \} \subset S$, where player $B$ does not use the action $B_2$, and
	\item  $S_{y_3} = \{ (x,y_1,y_2): \; y_1+y_2=1 \} \subset S$, where player $B$ does not use the action $B_3$.
\end{itemize}

Note that the Nash equilibria for the original coordination game are $(A_1,B_1)=(1,1,0)$ and $(A_2,B_2)=(0,0,1)$. The mixed strategy Nash equilibrium for the original game is also present in the flow invariant space $S_{y_3}$ and is given by
$$
(x,y_1, 1-y_1) = (p, q, 1-q) = (\dfrac{b_{2221}}{b_{2221}+b_{1112}}, \dfrac{a_{2212}}{a_{2212}+a_{1121}},  \dfrac{a_{1121}}{a_{2212}+a_{1121}}).
$$
Up to one further equilibrium may be present, when player $B$ abstains from using either action $B_1$ or action $B_2$. These correspond to
\begin{itemize}
	\item  either an equilibrium in the flow invariant plane $S_{y_1}$ for $b_{2223}b_{1312}>0$ and $a_{1323}>0$ given by
	$$
	(x,0,y_2)=\left(\dfrac{b_{2223}}{b_{2223}+b_{1312}}, 0, \dfrac{a_{1323}}{a_{1323}+a_{2212}}\right);
	$$
	
	\item  or an equilibrium in the flow invariant plane $S_{y_2}$ for $b_{2123}b_{1311}>0$ and $a_{1323}<0$ given by
	$$
	(x,y_1,0)=\left(\dfrac{b_{2123}}{b_{2123}+b_{1311}}, \dfrac{a_{1323}}{a_{1323}+a_{2111}}, 0\right).
	$$
\end{itemize}

\begin{lemma}
The dynamics of \eqref{eq:ODE} have at least 7 and at most 8 equilibria. 
There are no equilibria in the interior of $S$.
\end{lemma}

\begin{proof}
The proof is given in the detailed calculations in Appendix \ref{app:equilibria}. 
\end{proof}

The detail in Appendix \ref{app:equilibria} informs precisely about the values of the payoffs that allow for exactly 7 or exactly 8 equilibria in the dynamics of \eqref{eq:ODE}. The equilibria for the dynamics include the Nash equilibria determined above.

Table~\ref{tbl:eigen} provides the eigenvalues and eigenvectors of the Jacobian matrix for the ODEs \eqref{eq:ODE} at the corner equilibria. The detailed calculations are provided in Appendix \ref{app:Jacobian}.

\begin{table}
\begin{center}
\begin{tabular}{|c|c|c|}
\hline
equilibrium & eigenvalues & eigenvectors \\
\hline
  & $a_{13}-a_{23}$ & $(1,0,0)$ \\
$(0,0,0) = (A_2,B_3)$ & $b_{21}-b_{23}$ & $(0,1,0)$ \\
  & $b_{22}-b_{23}$ & $(0,0,1)$ \\
\hline
 & $-(a_{13}-a_{23})$ & $(1,0,0)$ \\
 $(1,0,0) = (A_1,B_3)$ & $b_{11}-b_{13}$ & $(0,1,0)$ \\
  & $b_{12}-b_{13}$ & $(0,0,1)$ \\
\hline
 & $a_{11}-a_{21} > 0$ & $(1,0,0)$ \\
 $(0,1,0) = (A_2,B_1)$ & $-(b_{21}-b_{23})$ & $(0,1,0)$ \\
  & $b_{22}-b_{21} > 0$ & $(0,1,-1)$ \\
\hline
 & $a_{12}-a_{22} < 0$ & $(1,0,0)$ \\
 $(0,0,1) = (A_2,B_2)$ & $-(b_{22}-b_{21}) < 0$ & $(0,1,-1)$ \\
  & $-(b_{22}-b_{23})$ & $(0,0,1)$ \\
\hline
 & $-(a_{11}-a_{21}) < 0$ & $(1,0,0)$ \\
 $(1,1,0) = (A_1,B_1)$ & $-(b_{11}-b_{13})$ & $(0,1,0)$ \\
  & $-(b_{11}-b_{12}) < 0$ & $(0,1,-1)$ \\
\hline
 & $-(a_{12}-a_{22}) > 0$ & $(1,0,0)$ \\
 $(1,0,1) = (A_1,B_2)$ & $b_{11}-b_{12} > 0$ & $(0,1,-1)$ \\
  & $-(b_{12}-b_{13})$ & $(0,0,1)$ \\
\hline
\end{tabular}
\end{center}
\caption{\small Information concerning the eigenvalues and eigenvectors of the Jacobian matrix for the ODEs \eqref{eq:ODE} at the corner equilibria. When the sign of the eigenvalues can be determined by the nature of the game, this is indicated.}
\label{tbl:eigen}
\end{table}

It is clear from Table~\ref{tbl:eigen} that at most one of the new  equilibria $(A_1,B_3)$ and $(A_2,B_3)$ is stable. It is also clear that the equilibria $(A_1, B_2)$ and $(A_2, B_1)$, also present in the original coordination game, are both unstable. And finally, the equilibria $(A_1, B_1)$ and $(A_2, B_2)$, both stable in the dynamics of the original coordination game, can now be independently either unstable (saddles) or stable (sinks). If these are both stable we replicate the need for an additional criterion for determining which equilibrium should be chosen. If however one, and only one of these is stable, then this may be the likely outcome of the modified game. This is addressed in Section~\ref{sec:choice}. Note that the original $2\times 2$ coordination game persists in $S_{y_3}$.

\subsection{Dynamics in flow-invariant subspaces}\label{sec:dyn-subspaces}

We now describe the dynamics in the flow invariant subspaces $S_{y_1}$, $S_{y_2}$, and $S_{y_3}$ of $S$ starting with the latter.

In $S_{y_3}$ the vector field \eqref{eq:ODE} becomes, dropping the index $1$ by writing $y\equiv y_1$, and $y_2=1-y$:
\begin{equation}
\left\{ 
\begin{aligned}
\dot{x} & = x(1-x)\left[ a_{1222} - (a_{1222}+a_{2111})y  \right]  \\
& \mbox{} \\
\dot{y} & =y(1-y) \left[ -b_{2221} + (b_{2221}+b_{1112})x \right] 
\end{aligned}
\right. .
\label{eq:Sy}
\end{equation}
These equations are as equations (2.1) in Hofbauer \cite{Hof1996} by identifying
$$
\begin{array}{ccc}
a=a_{1222}=a_{12}-a_{22} < 0 & \mbox{} & b=a_{2111}=a_{21}-a_{11} < 0 \\
c=b_{2221}=b_{22}-b_{21} > 0 & \mbox{} & d=b_{1112}=b_{11}-b_{12} > 0.
\end{array}
$$
Given the Hamiltonian 
$H(x,y)=c \log{x} + d \log {(1-x)} +a \log{y} + b \log{(1-y)}$ equations
\eqref{eq:Sy} can be written as
$$
\left\{ 
\begin{aligned}
\dot{x} & = P(x,y) \dfrac{\partial H}{\partial y}  \\
& \mbox{} \\
\dot{y} & = - P(x,y) \dfrac{\partial H}{\partial x}
\end{aligned}
\right. 
$$
for $P(x,y)=x(1-x)y(1-y)$. It is straightforward to see, given the signs of $a$, $b$, $c$, and $d$ above, that the mixed equilibrium $(p,q,1-q) \in S_y$ is a saddle. For initial conditions in this plane we recover the original coordination game.

\medskip

In $S_{y_1}$ the vector field \eqref{eq:ODE} becomes
$$
\left\{ 
\begin{aligned}
\dot{x} & = x(1-x)\left[ a_{1323} - (a_{1323}+a_{2212})y_2  \right]  \\
& \mbox{} \\
\dot{y}_2 & =y_2(1-y_2) \left[ -b_{2322} + (b_{2322}+b_{1213})x \right] 
\end{aligned}
\right. ,
$$
so that, a Hamiltonian function exists with
$$
\begin{array}{ccc}
a=a_{1323}=a_{13}-a_{23}  & \mbox{} & b=a_{2212}=a_{22}-a_{12}>0  \\
c=b_{2322}=b_{23}-b_{22}  & \mbox{} & d=b_{1213}=b_{12}-b_{13}.
\end{array}
$$
When an equilibrium exists in $S_{y_1}$ we have $a>0$ and $cd>0$. If $c,d>0$ then the equilibrium is a center, otherwise, it is a saddle. When the interior equilibrium in $S_{y_1}$ is a saddle there are two Nash equilibria on the boundary of $S_{y_1}$: $(A_1,B_1)$ and $(A_1,B_3)$. If no equilibrium exists in $S_{y_1}$ then the dynamics are determined by that along the edges that constitute the boundary of $S_{y_1}$.

Similar calculations show that the vector field  \eqref{eq:ODE} restricted to $S_{y_2}$ can be written in the same form but with
$$
\begin{array}{ccc}
a=a_{1323}=a_{13}-a_{23}  & \mbox{} & b=a_{2111}=a_{21}-a_{11}<0  \\
c=b_{2322}=b_{23}-b_{21}  & \mbox{} & d=b_{1113}=b_{11}-b_{13}.
\end{array}
$$
When an equilibrium exists in $S_{y_2}$ we have $a<0$ and $cd>0$. If $c,d<0$ then the equilibrium is a center, otherwise, it is a saddle. In this latter case, again we have two Nash equilibria on the boundary: $(A_2,B_2)$ and $(A_2,B_3)$. If no equilibrium exists in $S_{y_2}$ again the dynamics are determined by that along the edges that consitute the boundary of $S_{y_2}$.

The dynamics in the subspaces contained in the vertical planes given by $\{x=0\}$ and $\{x=1 \}$ is completely determined by the dynamics along the edges limiting these subspaces.

\section{Learning to coordinate}\label{sec:choice}

We examine how the existence of the new action $B_3$ for player $B$ can determine the outcome of the game (except for initial conditions in flow invariant subspaces, discussed above). 
In order to focus our analysis, we assume that $(A_2, B_2)$ is the payoff dominant equilibrium for player $B$. 

The study of the dynamics under any combination of the relative magnitude of the parameters is possible with the information provided in the appendices. However, our purpose is rather to prove the existence of certain reasonable choices of the payoffs that lead to the desirable dynamics. 

In order to achieve coordination on $(A_2, B_2)$, player $B$ learns to play a new action only if this ensures that $(A_2,B_2)$ is a stable Nash equilibrium. This means that $b_{22}>b_{23}$ to ensure stability (see Table~\ref{tbl:eigen}). 
Furthermore, player $B$ learns to play a new action if the payoff for this action is improved against at least one of the choices of player $A$. This corresponds to $b_{13}>b_{11}$. This means that player $B$ receives a higher payoff when using the new action $B_3$ against player $A$'s choice of $A_1$ but a lower payoff against $A_2$.
The only Nash equilibrium for the extended game is $(A_2, B_2)$ if we add the condition $a_{23}>a_{13}$.  Concerning player $A$, we are assuming that their choice of $A_2$ gives a better payoff than $A_1$ against $B_3$. This is a reasonable assumption if $(A_2, B_2)$ is payoff dominant for player $A$ as well, or it may be induced by a planner.

Since the existence of a unique Nash equilibrium which is stable is not sufficient to guarantee that there is convergence to this equilibrium, we prove that it is so under the above assumptions on the values of the payoffs. We formalise the assumptions as the following open condition:
\paragraph{Assumption (A)}
\begin{itemize}
	\item[(A1)]  $b_{22}>b_{23}$;
	\item[(A2)]  $b_{13}>b_{11}$;
	\item[(A3)]  $a_{23}>a_{13}$.
\end{itemize}
Note that (A3) ensures that no equilibrium exists in the interior of $S_{y_1}$. It follows from (A2) that an equilibrium exists in the interior of $S_{y_2}$ if and only if $b_{21}>b_{23}$.
See Figure~\ref{fig:PhaseDiagram_planes} for the dynamics in $S_{y_1} \cup S_{y_2}$ under Assumption (A). The phase diagram in $S_{y_3}$ is the usual for the $2 \times 2$ coordination game where a saddle in the interior of $S_{y_3}$ has an unstable manifold that connects to the pure equilibria $(A_1,B_1)$ and $(A_2,B_2)$. These are both stable for the dynamics restricted to $S_{y_3}$, recovering the original $2\times 2$ coordination game.

\begin{figure}
\centering
    \subfigure{
    \begin{tikzpicture}[xscale=1.1, yscale=1.1]
      \put (-60,40) {$S_{y_1}$}
      \put (45,-25) {$S_{y_2}$}
      \put (-80,-70) {\begin{tabular}{cl}
      $\mathbf{(a)}$ &
      $b_{21}-b_{23}>0$
      \end{tabular}};
\coordinate (1) at (0,0);
\coordinate (2) at (3,0);
\put (105,0,0) {$y_1$};
\coordinate (3) at (0,2.5);
\put (1,85,0) {$y_2$};
\coordinate (4) at (-2.5,-1.4);
\put (-90,-40) {$x$};
\coordinate (5) at (2,0);
\coordinate (6) at (0,1.5);
\coordinate (7) at (-1.6,-.9);
\coordinate (8) at (.65,-.9);
\coordinate (9) at (-1.65,.68);
\path (1) edge[-to] (2);
\path (1) edge[-to] (3);
\path (1) edge[-to] (4);
\filldraw[black] (1) circle (1pt);
\filldraw[black] (2,0) circle (1pt);
\filldraw[black] (0,1.5) circle (1pt) node[right] {\small$(A_2,B_2)$};
\filldraw[black] (-1.6,-.9) circle (1pt);
\filldraw[black] (.65,-.9) circle (1pt) node[below] {\small$(A_1,B_1)$};
\filldraw[black] (-1.65,.68) circle (1pt);
\path (5) edge[middlearrow={>}{left}{}] (8);
\path (1) edge[middlearrow={>}{left}{}] (5);
\path (8) edge[middlearrow={>}{left}{}] (7);
\path (7) edge[middlearrow={>}{left}{}] (1);
\path (9) edge[middlearrow={>}{left}{}] (7);
\path (9) edge[middlearrow={>}{left}{}] (6);
\path (1) edge[middlearrow={>}{left}{}] (6);
\path (9) edge[out=300,in=270,middlearrow={>}{}{}] (6);
\filldraw[black] (.3,-.45) circle (1pt);
\draw[->] (.3,-.15) arc (90:45:.4 and .3);
\draw(.3,-.45) ellipse (.4cm and .3cm);
  \end{tikzpicture}
    }
    \vspace{1cm}
    \subfigure{
    \begin{tikzpicture}[xscale=1.1, yscale=1.1]
      \put (-60,40) {$S_{y_1}$}
      \put (45,-25) {$S_{y_2}$}
      \put (-80,-70) {\begin{tabular}{cl}
      $\mathbf{(b)}$ &
      $b_{21}-b_{23}<0$
       \end{tabular}};
\coordinate (1) at (0,0);
\coordinate (2) at (3,0);
\put (105,0,0) {$y_1$};
\coordinate (3) at (0,2.5);
\put (1,85,0) {$y_2$};
\coordinate (4) at (-2.5,-1.4);
\put (-90,-40) {$x$};
\coordinate (5) at (2,0);
\coordinate (6) at (0,1.5);
\coordinate (7) at (-1.6,-.9);
\coordinate (8) at (.65,-.9);
\coordinate (9) at (-1.65,.68);
\path (1) edge[-to] (2);
\path (1) edge[-to] (3);
\path (1) edge[-to] (4);
\filldraw[black] (1) circle (1pt);
\filldraw[black] (2,0) circle (1pt);
\filldraw[black] (0,1.5) circle (1pt) node[right] {\small$(A_2,B_2)$};
\filldraw[black] (-1.6,-.9) circle (1pt);
\filldraw[black] (.65,-.9) circle (1pt) node[below] {\small$(A_1,B_1)$};
\filldraw[black] (-1.65,.68) circle (1pt);
\path (5) edge[middlearrow={>}{left}{}] (8);
\path (5) edge[middlearrow={>}{left}{}] (1);
\path (8) edge[middlearrow={>}{left}{}] (7);
\path (7) edge[middlearrow={>}{left}{}] (1);
\path (9) edge[middlearrow={>}{left}{}] (7);
\path (9) edge[middlearrow={>}{left}{}] (6);
\path (1) edge[middlearrow={>}{left}{}] (6);
\path (9) edge[out=300,in=270,middlearrow={>}{}{}] (6);
\path (5) edge[out=200,in=270,middlearrow={>}{}{}] (1);
  \end{tikzpicture}
    }
\caption{The phase diagram in $S_{y_1} \cup S_{y_2}$ under Assumption (A), depending on the sign of $b_{21}-b_{23}$.}
\label{fig:PhaseDiagram_planes}
\end{figure}
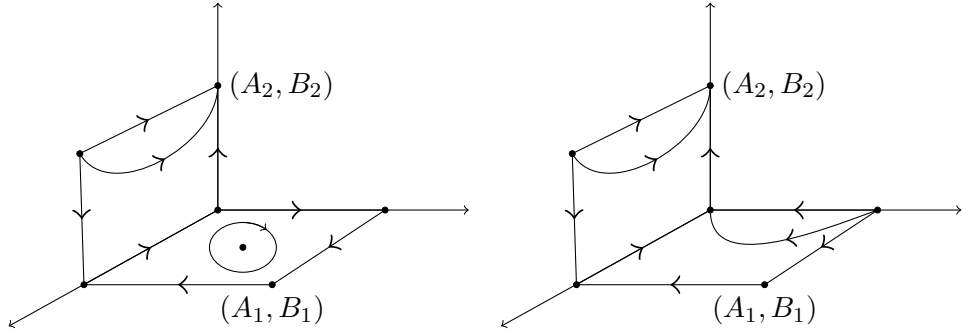

The next result establishes that the $\omega$-limit of any initial condition in the interior of $S$ is $(A_2,B_2)$, thus showing that by having an additional action at their disposal player $B$ can force their most desired Nash equilibrium in the original $2 \times 2$ game to be the outcome of the extended $2 \times 3$ game. The following quantity (see Appendix~\ref{app:Jacobian}) is useful:
$$
Q= (b_{22}-b_{21})(b_{11}-b_{13})+(b_{11}-b_{12})(b_{21}-b_{23}).
$$

\begin{theorem}\label{th:choice}
Suppose that Assumption (A) holds and that $Q<0$.
Then the $\omega$-limit of any initial condition in the interior of $S$ is $(A_2,B_2)$.
\end{theorem}

\begin{proof}
Under the hypotheses above, the only locally asymptotically stable equilibrium for the dynamics of \eqref{eq:ODE} is $(A_2,B_2)$. 
We claim that none of the flow invariant planes is attracting. Using Appendix~\ref{app:Jacobian}, the equilibrium in the interior of $S_{y_2}$, when it exists, is repelling towards the interior of $S$ if and only if the eigenvalue in the direction of $(0,0,1)$ is positive, that is, $Q<0$. The equilibrium in the interior of $S_{y_3}$ is a saddle in $S_{y_3}$. It is repelling towards the interior of $S$ if the third eigenvalue is positive. This occurs when the determinant of the Jacobian matrix at this equilibrium is negative, that is, $Q<0$. Therefore $Q<0$ ensures that initial conditions, near the interior equilibria that may exist in some of the invariant planes, move to the interior of $S$.

It now suffices to show that the $\omega$-limit set of an initial condition in the interior of $S$ cannot be a more complicated invariant set in the interior of $S$.

Theorem A in Sela \cite{Sela} guarantees convergence to equilibrium in $2 \times 3$ games under fictitious play. Using the results of Hofbauer {\em et al.} \cite{HSV}, under the equivalent replicator dynamics there is convergence to an invariant object whose time-average is an equilibrium. Thus, under replicator dynamics in a $2 \times 3$ game solutions converge to either a periodic orbit or an equilibrium. 

Suppose there exists an $\omega$-limit set of an interior point of $S$ that is a periodic orbit. Then the time-average of this periodic orbit is an equilibrium. This is impossible since there are no equilibria in the interior of $S$ (see Appendix~\ref{app:equilibria}) for the dynamics of \eqref{eq:ODE}, finishing the proof.
\end{proof}

If $b_{21}<b_{23}$ then no equilibrium exists in the interior of $S_{y_2}$. In this case, (A1) and (A2) ensure that $Q<0$. Otherwise, $Q$ may have either sign depending on the relative magnitude of the deviations among the payoffs. We have $Q<0$ in case $b_{21}>b_{23}$ when 
\begin{equation}\label{eq:deviations}
\frac{b_{21}-b_{22}}{b_{23}-b_{21}} > \frac{b_{11}-b_{12}}{b_{13}-b_{11}}
\end{equation}
meaning that the ratio of the deviation among payoffs for player $B$ is bigger when player $A$ is using $A_2$ than when player $A$ is using $A_1$. 

\subsection{Illustrative examples}\label{subsec:examples}

We illustrate our results with the Stag Hunt and Battle of the Sexes games presented at the beginning of Section~\ref{sec:intro}.

\paragraph{Stag Hunt:}
Suppose that a third option is presented to player $B$ for them to hunt. Let us say this is Partridge. Leaving the payoffs for the choices of Hare and Stag unchanged, assume generic payoffs for the choice of Partridge by player $B$ for both players as in the beginning of Section~\ref{sec:CoordinationGames}.

Assumption (A) requires that $10 > b_{23}$, $b_{13} > 5$ and $a_{23} > a_{13}$. This means that player $B$ prefers Partridge to Hare but not to Stag, when player $A$ makes the same choice. The fact that player $B$ values the new action more than one of the original actions makes the use of the new action credible in the eyes of player $A$. On the other hand, player $A$ is happier having attempted to hunt Stag than Hare. Although this may seem implausible since player $A$ cannot successfully hunt Stag on their own, it may be that they get some satisfaction out of, or compensation from, the result of player $B$. This may be the case if a planner were involved.\footnote{Anchors and nudges as introduced by Kahneman \cite{Kah2011} and Thaler \cite{Tha2008}, respectively, may serve this purpose. Although this opens an interesting further path for research, it is beyond the scope of the present article.}

We can calculate $Q$ to find $Q=42-2b_{13}-4b_{23}$.  Using Theorem~\ref{th:choice}, gives a range of values for $b_{13}$ and $b_{23}$ to determine the outcome of the game. If $8 < b_{23} < 10$ then $Q<0$ is satisfied for any value of $b_{13}$ satisfying Assumption (A). The inequality $b_{23} > b_{21}$ reinforces the credibility of the use of the new action. If $b_{23} < 8$ then $Q<0$ becomes binding.

\paragraph{Battle of the Sexes:}
Suppose that player $B$ has found an interesting film and can choose Cinema. Again the payoffs for the original action remain unchanged.

Assumption (A) means that player $B$ is happier choosing Cinema alone over Opera alone, but not choosing Opera when player $A$ does the same. Player $A$ gets less satisfaction when choosing Football than Opera when player $B$ chooses Cinema. Some external intervention may introduce a compensation to achieve this.

We find that $Q = 70-10b_{13}-7b_{23}$. This condition is binding if $b_{23}<0$.

\section{Concluding remarks}

We address the issue of attaining coordination at a preferred Nash equilibrium in a $2\times 2$ game. We show that coordination on the desired outcome can be achieved by embedding the original game in a modified game where one of the players can use an additional action. 
Thus, by preventing the least desired Nash equilibrium from being chosen, we address the problem of coordination failure. Miscoordination does not occur since neither of the mixed equilibria is stable in the modified game.

The learning dynamics are of replicator type and describe the players' choices as time goes by. Our Assumption (A) is mild and realistic in practice since it is an open condition. Another level of learning is introduced in the modification of the $2\times 2$ game: we assume that one of the players can learn to use a new action. This is realistic in that agents not only improve the use of their skills but acquire new ones. This new action can become available through innovation and is learnt by one of the players. That the newly developed action does not become available to both players is also realistic in noncooperative games: the development of new technologies by firms or countries is usually a well-kept secret. However, the existence of  a technological improvement is frequently leaked and its availability for one of the players becomes common knowledge.

The purpose of the third action is not to be used, although it is possible to find payoff values that make using the new action the only Nash equilibrium of the modified game. Instead, the existence of the additional action can be seen as an anchor, in the sense of Kahneman \cite{Kah2011} or a nudge, in the sense of Thaler \cite{Tha2008}, promoting the choice of the optimal Nash equilibrium of the original coordination game. Desirable anchors or nudges can be created by a planner by making an action available to one but not both players (e.g., a planner can make a route available to bicycles but not to cars). 

If both players were to learn to use the new action then we fall into a $3\times 3$ coordination game with the persistence of the coordination problem. See Lahkar and Seymour \cite{LahSey2013} or Arigapudi \cite{Ari2022} showing that the lack of coordination persists. The additional action can thus be seen to work as an advantage for the player who can use it or threat for the player who cannot.

Regardless of the possible interpretations of the existence of an additional action for just one of the players, this article provides a novel approach to achieving coordination that hopefully can be tested in experiments.


\paragraph{Acknowledgements:} The author was partially supported by CMUP, member of LASI, which is financed by national funds through FCT Funda\c{c}\~ao para a Ci\^encia e a Tecnologia, I.P., under the projects UIDB/00144/2020 and UIDP/00144/2020. I am grateful to J.\ Correia-da-Silva and to P.\ Gothen for their comments on a preliminary version. The present exposition of my ideas benefitted also from the comments from R.\ Tierney, R.\ Treibich and H.\ Zhang during the Microeconomics Group Seminar at SDU, Denmark. This visit was supported by an Erasmus+ Staff Training Mobility Grant.


\appendix

\section{Equilibria for the dynamics} \label{app:equilibria}

Given the ODEs \eqref{eq:ODE} it is easy to see that the corners of the state space $S$ are equilibria. These are 
$(x,y_1,y_2)=(0,0,0)$, $(x,y_1,y_2)=(1,0,0)$, $(x,y_1,y_2)=(0,1,0)$,  $(x,y_1,y_2)=(0,0,1)$, $(x,y_1,y_2)=(1,1,0)$, and 
$(x,y_1,y_2)=(1,0,1)$. We look for equilibria on the faces of $S$ by solving the following systems of equations:

\paragraph{$\bullet$ when $x=0$ and $y_1y_2 \neq 0$:}
\[
\begin{cases}
b_{2123}(1-y_1) - b_{2223}y_2=0 \\
b_{2223}(1-y_2) - b_{2123}y_2=0
\end{cases}
\Leftrightarrow
\begin{cases}
b_{21}-b_{23} - b_{22} + b_{23}=0 \\
b_{2223}(1-y_2) - b_{2123}y_2=0
\end{cases}
\]
which has no solution since $b_{22}-b_{21}>0$.

\paragraph{$\bullet$ when $x=1$ and $y_1y_2 \neq 0$:}
\[
\begin{cases}
b_{1113}(1-y_1) - b_{1213}y_2=0 \\
b_{1213}(1-y_2) - b_{1113}y_2=0
\end{cases}
\Leftrightarrow
\begin{cases}
b_{11}-b_{13} - b_{12} + b_{13}=0 \\
b_{1213}(1-y_2) - b_{1113}y_2=0
\end{cases}
\]
which has no solution since $b_{11}-b_{12}>0$.

\paragraph{$\bullet$ when $y_1=0$ and $xy_2 \neq 0$:}
\[
\begin{cases}
a_{1122}y_2 + a_{1323}(1-y_2)=0 \\
b_{2223}x + b_{1213}(1-x)=0
\end{cases}
\Leftrightarrow
\begin{cases}
y_2=\dfrac{a_{1323}}{a_{1323}+a_{2212}} \\
\mbox{} \\
x=\dfrac{b_{2223}}{b_{2223}+b_{1312}}
\end{cases}
\]
These coordinates correspond to a point in state space $S$ if and only if both $a_{1323}a_{2212}>0$ and $b_{2223}b_{1312}>0$. Given that $a_{2212}>0$ the first inequality reduces to $a_{1323}>0$.

\paragraph{$\bullet$ when $y_2=0$ and $xy_1 \neq 0$:}
\[
\begin{cases}
a_{1121}y_1 + a_{1323}(1-y_1)=0 \\
b_{1113}x + b_{2123}(1-x)=0
\end{cases}
\Leftrightarrow
\begin{cases}
y_1=\dfrac{a_{1323}}{a_{1323}+a_{2111}} \\
\mbox{} \\
x=\dfrac{b_{2123}}{b_{2123}+b_{1311}}
\end{cases}
\]
These coordinates correspond to a point in state space $S$ if and only if both $a_{1323}a_{2111}>0$ and $b_{2123}b_{1311}>0$. Since $a_{2111}<0$ the first inequality is equivalent to $a_{1323}<0$. 

Hence, only one of the two latter equilibria occurs at a time. It can be that neither equilibrium exists.

\paragraph{$\bullet$ when $xy_1y_2 \neq 0$:}
from the last two equations in \eqref{eq:ODE}, we obtain
\begin{align*}
&
\begin{cases}
\left[ b_{2123} - (b_{2123}+b_{1311})x \right](1-y_1) - \left[ b_{2223}- (b_{2223}+b_{1312})x \right]y_2=0 \\
\left[ (b_{2223} - (b_{2223} +b_{1312})x \right](1-y_2) - \left[ (b_{2123} - (b_{2123}+b_{1311})x \right]y_1=0
\end{cases}
\Leftrightarrow \\
\Leftrightarrow &
\begin{cases}
x=p=\dfrac{b_{2221}}{b_{2221}+b_{1112}} \\
\mbox{} \\
y_2=1-y_1
\end{cases}.
\end{align*}
Using the first equation in \eqref{eq:ODE}, we obtain 
$$
y_1=q=\dfrac{a_{2212}}{a_{2212}+a_{1121}}.
$$
Therefore, $y_3=0$ and there are no equilibria in the interior of $S$.

\section{The Jacobian matrix at the equilibria} \label{app:Jacobian}

The entries of the Jacobian matrix for system \eqref{eq:ODE} are as follows
\begin{align*}
J_{11} & = (1-2x)\left[ a_{1121}y_1 + a_{1222}y_2 + a_{1323}(1-y_1-y_2) \right] \\
J_{12} & = x(1-x)(a_{1121} - a_{1323}) \\
J_{13} & = x(1-x)(a_{1222} - a_{1323}) \\
\mbox{}& \\
J_{21} & = y_1\left[ (b_{1113} - b_{2123})(1-y_1) - (b_{1213} - b_{2223})y_2 \right] \\
J_{22} & = (1-2y_1)\left[ b_{2123}(1-x) + b_{1113}x \right] - y_2\left[ b_{2223}(1-x) + b_{1213}x \right] \\
J_{23} & = -y_1\left[ b_{2223}(1-x) + b_{1213}x \right] \\
\mbox{}& \\
J_{31} & = y_2\left[ (b_{1213} - b_{2223})(1-y_2) -(b_{1113} - b_{2123})y_1 \right] \\
J_{32} & = -y_2\left[ b_{2123}(1-x) + b_{1113}x \right] \\
J_{33} & = (1-2y_2)\left[ b_{2223}(1-x) + b_{1213}x \right] -y_1\left[ b_{2123}(1-x) + b_{1113}x \right]
\end{align*}

Replacing the coordinates by their value at each equilibrium, we obtain
\begin{align*}
J_{(0,0,0)=(A_2,B_3)}& = \begin{bmatrix}
a_{1323} & 0 & 0 \\[0.1cm]
0 & b_{2123} & 0 \\[0.1cm]
0 & 0 & b_{2223}
\end{bmatrix}, 
& J_{(1,0,0)=(A_1,B_3)}& = \begin{bmatrix}
-a_{1323} & 0 & 0 \\[0.1cm]
0 & b_{1113} & 0 \\[0.1cm]
0 & 0 & b_{1213}  
\end{bmatrix}, \\[0.4cm]
J_{(0,1,0)=(A_2,B_1)}& = \begin{bmatrix}
a_{1121} & 0 & 0 \\[0.1cm]
0 & -b_{2123} & -b_{2223} \\[0.1cm]
0 & 0 & b_{2221}
\end{bmatrix},
&  J_{(0,0,1)=(A_2,B_2)}& = \begin{bmatrix}
a_{1222} & 0 & 0 \\[0.1cm]
0 & b_{2122} & 0 \\[0.1cm]
0 & -b_{2123} & -b_{2223}
\end{bmatrix}, \\[0.4cm]
J_{(1,1,0)=(A_1,B_1)}& = \begin{bmatrix}
-a_{1121} & 0 & 0 \\[0.1cm]
0 & -b_{1113} & -b_{1213} \\[0.1cm]
0 & 0 & b_{1211}
\end{bmatrix},
&  J_{(1,0,1)=(A_1,B_2)}& = \begin{bmatrix}
-a_{1222} & 0 & 0 \\[0.1cm]
0 & b_{1112} & 0 \\[0.1cm]
0 & -b_{1113} & -b_{1213}
\end{bmatrix}.
\end{align*}
All the matrices are either diagonal or triangular so that the eigenvalues can be read off the diagonal. In all cases, the eigenvector that is not a coordinate vector is $(0,1,-1)$.

At the equilibrium in the interior of $S_{y_1}$ we have
$$
J_{y_1} = \begin{bmatrix}
0 & J_{12} & J_{13} \\[0.1cm]
0 & J_{22} & 0 \\[0.1cm]
J_{31} & J_{32} & 0  
\end{bmatrix}.
$$
The eigenvalue corresponding to the direction of $(0,1,0)$, orthogonal to $S_{y_1}$ is $J_{22}$. This is
$$
J_{22}=\frac{b_{2123}b_{1312}+b_{1113}b_{2223}}{b_{2223}+b_{1312}}
$$
whose sign is that of either $Q=b_{21}b_{13}-b_{21}b_{12}+b_{23}b_{12}+b_{11}b_{22}-b_{11}b_{23}-b_{13}b_{22}$ or its symmetric, depending on the sign of $b_{2223}+b_{1312}$. We can write $Q$ as deviations for player B as follows
$$
Q= (b_{22}-b_{21})(b_{11}-b_{13})+(b_{11}-b_{12})(b_{21}-b_{23}).
$$
The eigenvalues whose eigenvectors are in $S_{y_1}$ are those of
$$
\begin{bmatrix}
0 &  J_{13} \\[0.1cm]
J_{31}  & 0  
\end{bmatrix}.
$$
This matrix has zero trace and the sign of its determinant is that of $-(b_{1312}+b_{2223})$. 

\medskip

At the equilibrium in the interior of $S_{y_2}$ we have
$$
J_{y_2} = \begin{bmatrix}
0 & J_{12} & J_{13} \\[0.1cm]
J_{21} & 0 & J_{23} \\[0.1cm]
0 & 0 & J_{33}   
\end{bmatrix}.
$$
The eigenvalue corresponding to the direction of $(0,0,1)$, orthogonal to $S_{y_2}$ is $J_{33}$, where
$$
J_{33} = -\frac{y_1 Q}{b_{2223}+b_{1312}} = -\frac{a_{1323}}{a_{1323}+a_{2111}} \frac{Q}{b_{2223}+b_{1312}}.
$$
It is clear that $J_{22}$ and $J_{33}$ have opposite signs.
The eigenvalues whose eigenvectors are in $S_{y_2}$ are those of
$$
\begin{bmatrix}
0 &  J_{12} \\[0.1cm]
J_{21}  & 0  
\end{bmatrix}.
$$
This matrix has zero trace and the sign of its determinant is that of $b_{1312}+b_{2223}$. 

\medskip

At the equilibrium in the interior of $S_{y_3}$ we have
\begin{equation}\label{eq:Jy}
J_{y_3} = \begin{bmatrix}
0 & (a_{1121}-a_{1323})p(1-p) & (a_{1222}-a_{1323})p(1-p) \\[0.1cm]
(b_{1112}+b_{2221})(1-q)q & -\beta q & -\alpha q  \\[0.1cm]
-(b_{1112}+b_{2221})(1-q)q & -\beta (1-q) & -\alpha (1-q)
\end{bmatrix},
\end{equation}
where 
\begin{align*}
\alpha = b_{1213}p + b_{2223}(1-p) \\
\beta = b_{1113}p + b_{2123}(1-p).
\end{align*}
Note that $\alpha -\beta = 0$. The determinant of $J_{y_3}$ has sign equal to that of $Q$.
\end{document}